\documentclass[lettersize,journal]{IEEEtran}
\usepackage[ruled,linesnumbered,lined,noend]{algorithm2e}
\usepackage[switch]{lineno}
\usepackage{cite}
\usepackage{graphicx}
\usepackage{subfigure}
\usepackage{pifont} 
\usepackage{url}
\usepackage{color}
\usepackage{booktabs}
\usepackage{multirow}
\usepackage{latexsym,bm,amsmath,amssymb}
\usepackage{array}
\usepackage{algpseudocode}
\usepackage{lscape}
\usepackage{comment}
\usepackage{balance}
\usepackage[colorlinks,linkcolor = blue,citecolor = blue,citecolor = blue,filecolor = blue,urlcolor=blue]{hyperref}
\usepackage[most]{tcolorbox}
\usepackage{ragged2e}
\usepackage{listings}
\usepackage{tabularx}
\usepackage{amsthm}
\usepackage{colortbl}
\usepackage{xcolor}

\definecolor{lightgray}{gray}{0.9}

\newcommand{\tool}{\textsf{VeriGround}\xspace}

\definecolor{lightgreen}{RGB}{232, 245, 233}   
\definecolor{mediumgreen}{RGB}{165, 214, 167}   
\definecolor{darkgreen}{RGB}{129, 199, 132}      

\newtheorem{definition}{Definition}
\newtheorem{proposition}{Proposition}

\newtcolorbox{findingbox}{
  colback=mygray,
  colframe=mygray,
  boxrule=0pt,
  left=4pt, right=4pt, top=2pt, bottom=2pt,
  fontupper=\bfseries,
  before skip=6pt,
  after skip=2pt
}

\usepackage{listings}
\definecolor{light-gray}{gray}{0.85}
  
\newcommand{\code}[1]{\colorbox{light-gray}{\texttt{#1}}}

\definecolor{mygray}{gray}{.9}

\begin{document}

\title{From Mirage to Grounding: Towards Reliable Multimodal Circuit-to-Verilog Code Generation}

\author{Guang Yang, Xing Hu\IEEEauthorrefmark{1}\thanks{Corresponding author: Xing Hu.}, Xiang Chen, and Xin Xia

\IEEEcompsocitemizethanks{
\IEEEcompsocthanksitem Guang Yang is with the State Key Laboratory of Blockchain and Data Security, Zhejiang University, Hangzhou, China, and also with the Hangzhou High-Tech Zone (Binjiang) Institute of Blockchain and Data Security, Hangzhou, China. 
Xing Hu and Xin Xia are with the State Key Laboratory of Blockchain and Data Security, Zhejiang University, Hangzhou, China. 
Xiang Chen is with the School of Artificial Intelligence and Computer Science, Nantong University, Nantong, China. 
E-mail: novelyg@outlook.com, xinghu@zju.edu.cn, xin.xia@acm.org, xchencs@ntu.edu.cn.
}

\thanks{Manuscript received April 19, 2020; revised August xx, xxxx.}}

\markboth{IEEE TRANSACTIONS ON Software Engineering,~Vol.~XX, No.~XX, XX~2026}%
{From Mirage to Grounding: Towards Reliable Multimodal Circuit-to-Verilog Code Generation}

\IEEEtitleabstractindextext{
\begin{abstract}
\justifying
Multimodal large language models (MLLMs) are increasingly used to translate visual artifacts into code, from UI mockups into HTML to scientific plots into Python scripts.
A circuit diagram can be viewed as a \emph{visual domain-specific language} for hardware: it encodes timing, topology, and bit-level semantics that are invisible to casual inspection yet safety-critical once fabricated in silicon.
Translating such diagrams into register-transfer-level (RTL) code therefore represents an extreme reliability test for vision-to-code generation.
We reveal a phenomenon we call \emph{Mirage}: replacing a circuit diagram with a blank image leaves Pass@$k$ unchanged or even higher, because models bypass the visual input and instead exploit identifier semantics in the \texttt{module\_header} to retrieve canonical RTL templates.
This constitutes a new, highly covert class of defect in AI-assisted code generation that directly undermines MLLMs' trustworthiness.
To quantify the effect, we construct \textsc{C2VEval} and evaluate eight MLLMs under a paired \emph{Normal/Anony} protocol in which \emph{Anony} mode anonymizes all identifiers in both the diagram and the module header; Anony-mode scores drop sharply across all models, confirming that high Normal-mode accuracy is largely a Mirage.
We then propose \tool (4B), trained with identifier anonymization, refusal augmentation, and D-ORPO (Decision-Focused ORPO) preference alignment that up-weights pivotal generate-or-refuse tokens.
\tool achieves Functional Pass@1 of 46.11\%/42.51\% (Normal/Anony) with a False Refusal Rate of only 1.20\%/0.00\%, while maintaining $\geq$92\% Refusal Rate on blank images.
With only 4B parameters, \tool performs on par with GPT-5.4 under Normal and significantly outperforms all baselines under Anony, confirming genuine visual grounding.
The evaluation methodology and training recipe generalize beyond hardware: \textsc{C2VEval}'s paired protocol can benchmark any vision-to-code pipeline, and \tool's anonymization-refusal-alignment triad offers a transferable solution for balancing hallucination and over-refusal in AI code generation.
\end{abstract}

}

\maketitle

\IEEEdisplaynontitleabstractindextext
\IEEEpeerreviewmaketitle

\section{Introduction}
\label{sec:intro}

Image-to-code generation has become a central paradigm in AI-assisted coding.
From UI mockups rewritten as HTML~\cite{feng2021auto, xiao2024prototype2code, gui2025uicopilot, zhou2025declarui, xiao2025interaction2code, wan2025divide}, to scientific plots reverse-engineered into reproducible Python scripts~\cite{wu2025plot2code, zhao2025chartcoder, ouyang2026dscodebench}, multimodal large language models (MLLMs) are shifting the front-end of software creation from purely textual specifications toward richer visual artifacts.
In each of these tasks, the visual input can be viewed as a \emph{visual domain-specific language (visual DSL)}: a UI mockup specifies layout and interaction semantics, a chart encodes data and rendering logic, and each must be faithfully translated into executable code.
Recently, this paradigm has been extended to hardware: MLLMs are used to translate \emph{circuit diagrams} into synthesizable register-transfer-level (RTL) code, a task we call \emph{circuit-to-Verilog code generation}~\cite{wong2024vgv, zhang2026mgemmv}.
A circuit diagram is arguably the demanding visual DSL, because it encodes timing, topology, and bit-level semantics that are invisible to casual inspection yet safety-critical once fabricated in silicon.

\begin{figure}[t]
\centering
\includegraphics[width=0.48\textwidth]{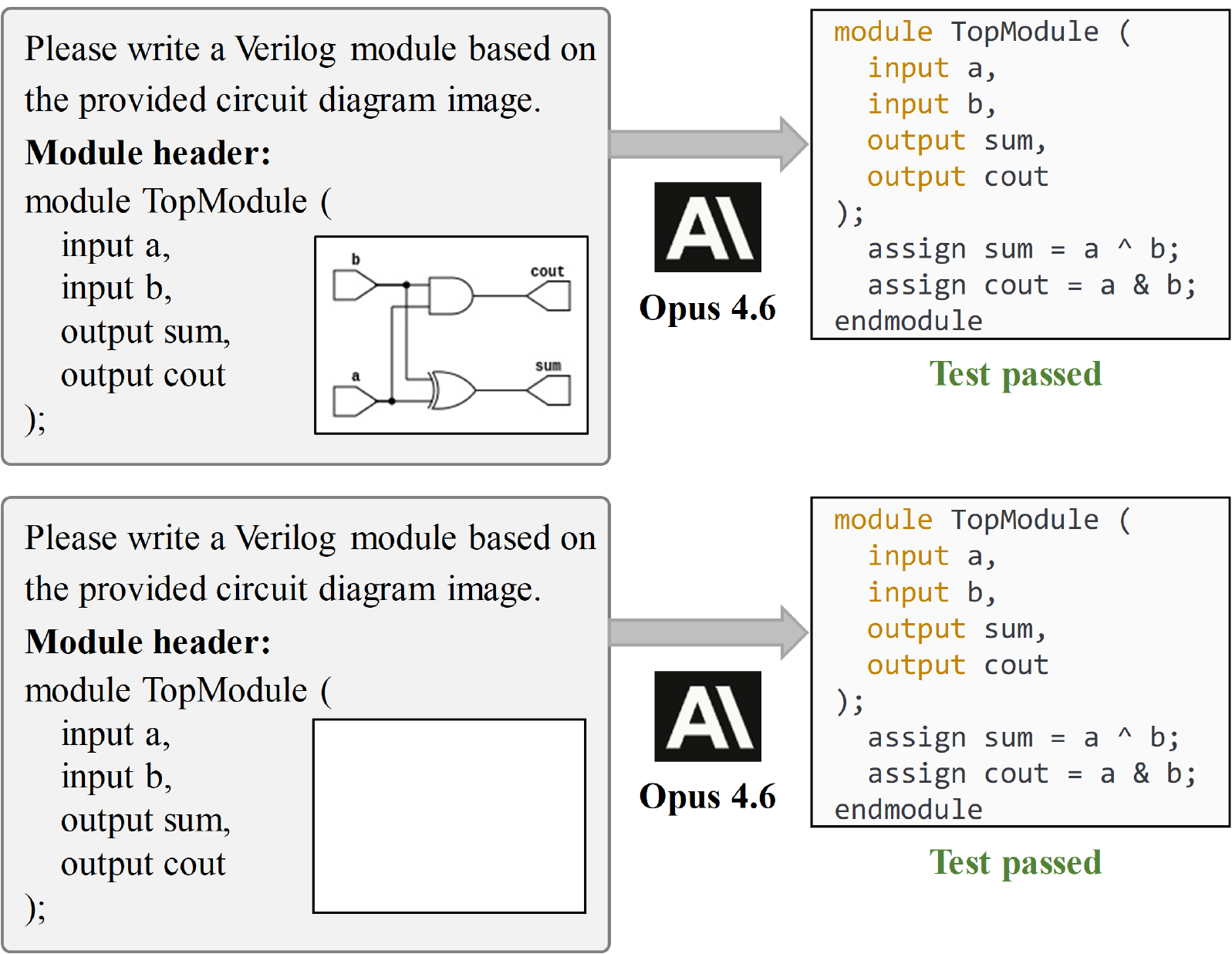}
\caption{Motivating example~1 of the \emph{Mirage} phenomenon. The model generates correct code regardless of whether the input contains the real circuit diagram or a blank image.}
\label{fig:example1}
\end{figure}
\begin{figure*}[t]
\centering
\includegraphics[width=0.8\textwidth]{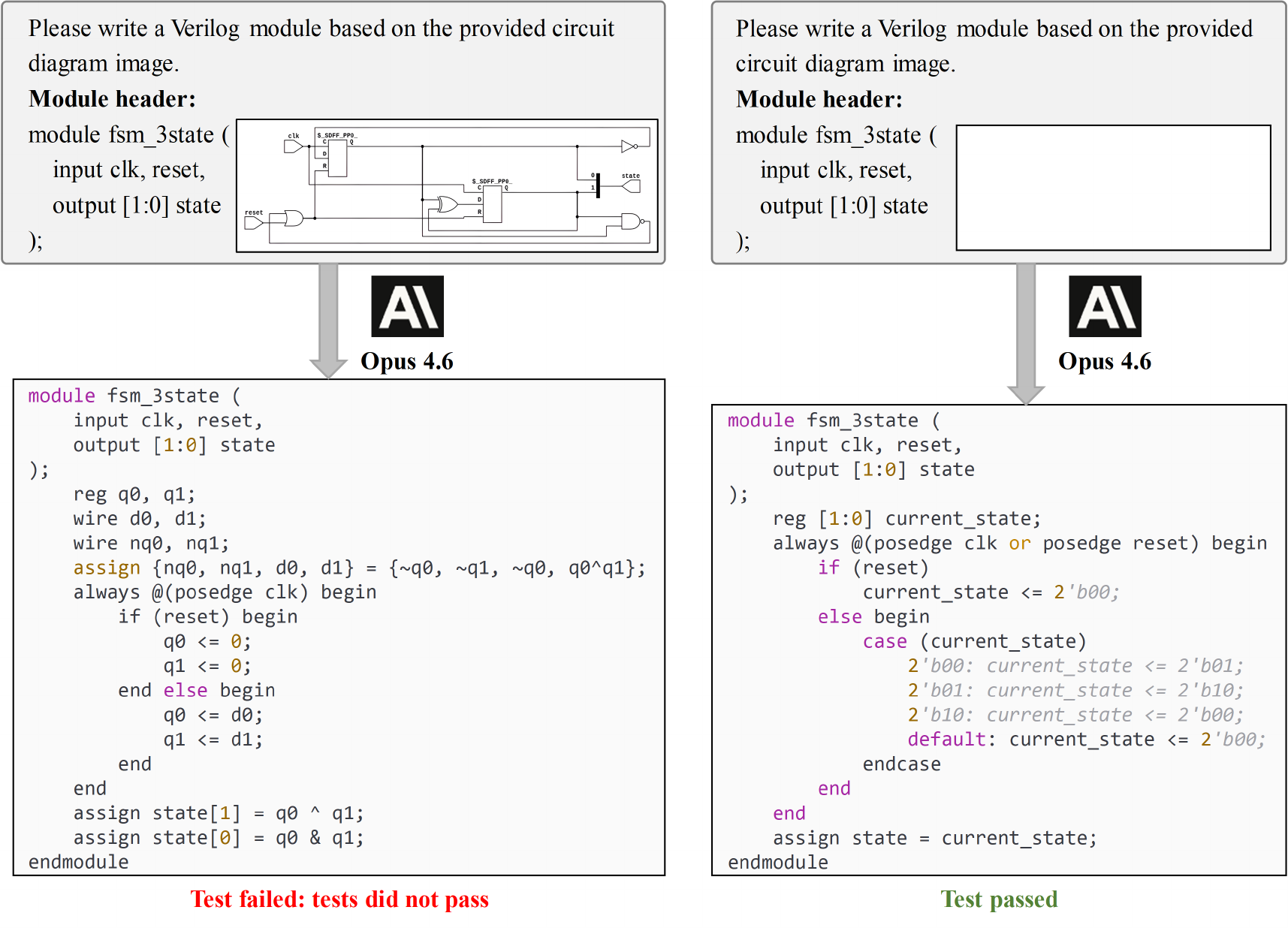}
\caption{Motivating example~2 of the \emph{Mirage} phenomenon. When the real circuit diagram is provided, the model generates incorrect Verilog that fails the testbench. When the diagram is replaced by a blank image while the \texttt{module\_header} is retained, the model instead produces correct code.}
\label{fig:example2}
\end{figure*}
  
While unreliable vision-to-code translation is a general concern across all visual DSLs, the consequences in the hardware domain are uniquely severe~\cite{yang2025large, pan2025survey}.
RTL sits at the very beginning of the silicon design flow and feeds downstream into synthesis, place-and-route, and fabrication; a single misread schematic can propagate silently through the entire toolchain and materialize as a silicon-level defect whose cost is orders of magnitude higher than a rendering or plotting error.
More broadly, if an AI code generator appears to ``understand'' a visual specification while actually bypassing it, the resulting code carries a covert correctness risk that conventional testing may not catch.
This directly threatens MLLMs' \emph{trustworthiness}, a core concern of the software engineering community.
Circuit-to-Verilog generation therefore serves as a rigorous proving ground: any reliability flaw exposed here is likely to manifest, in milder but equally insidious forms, in other vision-to-code pipelines.

Although recent studies have demonstrated the feasibility of using MLLMs to read circuit diagrams and generate Verilog code, a fundamental question remains largely unexplored: \emph{do existing MLLMs truly read circuit diagrams, or merely exploit textual shortcuts?}

Our motivating examples provide concrete evidence for this concern.
In both Fig.~\ref{fig:example1} and Fig.~\ref{fig:example2} the model under test is Opus~4.6, one of the current frontier code-generation models.
Following standard practice, the \texttt{module\_header} (module name, ports, and parameters) is provided as part of the prompt so that the generated interface matches the testbench and Pass@$k$ can be computed without name-mismatch failures.
Yet this seemingly innocuous input turns out to be a powerful textual shortcut.
In Fig.~\ref{fig:example1}, the model produces identical correct code for a half-adder whether the input contains the real circuit or a blank image: the module name \texttt{TopModule} is semantically vacuous, but the output ports \texttt{sum} and \texttt{cout} already reveal the target function, making the diagram redundant.
Fig.~\ref{fig:example2} is even more striking: for \texttt{fsm\_3state}, the model \emph{fails} with the real diagram yet \emph{succeeds} with a blank image, suggesting that the visual input can actively interfere with generation.
We term this failure mode the \emph{Mirage} phenomenon~\cite{asadi2026mirage}: high benchmark scores mask the fact that models rely on textual priors in the \texttt{module\_header} rather than genuinely grounding in the visual circuit topology.
Mirage constitutes a new, highly covert class of AI code-generation defect: the output may compile and even pass certain tests, yet it was never derived from the visual specification, leaving latent errors that surface only in untested scenarios.
A natural way to test this hypothesis is to anonymize the identifiers in both the \texttt{module\_header} and the circuit diagram, thereby stripping the semantic cues that enable such shortcuts.

Following this idea, we construct \textsc{C2VEval} (Circuit-to-Verilog Evaluation), a benchmark that samples problems from established Verilog code-generation benchmarks and renders each reference solution into a circuit diagram via \textsc{netlistsvg}\footnote{\url{https://github.com/nturley/netlistsvg}}, ensuring exact image-code correspondence.
\textsc{C2VEval} has two variants: \emph{Normal}, which retains the original identifiers in both the diagram and the module header, and \emph{Anony} (anonymized), which replaces all semantically loaded identifiers with positional placeholders and re-renders the diagram accordingly, preserving circuit topology while removing semantic cues.
We evaluate eight MLLMs, from 4B open-source models to frontier systems, under each variant in two modes: \emph{Original} (real diagram provided) and \emph{Mirage} (diagram replaced by a blank image, header retained).
The experimental results achieve the following three findings.
(i)~under \emph{Normal}, Mirage mode matches or exceeds Original on every model, suggesting that models largely bypass the circuit diagram and sometimes perform worse when it is present;
(ii)~under \emph{Anony}, the Mirage advantage reverses for seven of eight models, isolating identifier semantics in the header as the primary driver of Normal-mode performance;
and (iii)~genuine visual grounding accounts for only $\sim$8--9\% of samples, with the vast majority of tasks remaining unsolved once identifier shortcuts are removed.

Based on these findings, we first curate a large-scale circuit-diagram-to-Verilog dataset by mining high-quality Verilog projects from GitHub and converting each design into a circuit diagram.
We then introduce \tool (4B), a lightweight MLLM trained with three targeted interventions that together address the hallucination-vs-refusal trade-off common to all vision-to-code generators:
(i)~\emph{mixed supervised fine-tuning} on Original and Anony data to reduce reliance on semantically loaded identifiers;
(ii)~\emph{refusal augmentation} with blank-diagram and image-header-mismatch negatives, teaching the model to abstain when visual evidence is absent or inconsistent;
(iii)~\emph{D-ORPO (Decision-Focused ORPO) alignment}, which up-weights the first $K$ response tokens with a decision weight $\alpha$, concentrating the preference signal on the generate-or-refuse boundary and mitigating the over-refusal problem of standard ORPO~\cite{hong2024orpo}.

Extensive experiments demonstrate that \tool effectively alleviates the Mirage phenomenon.
With only 4B parameters, \tool achieves Functional Pass@1 of 46.11\% under \emph{Normal}, approaching GPT-5.4 (45.51\%) and surpassing GPT-4o (33.52\%) and MiMo-v2-omni (37.72\%); under \emph{Anony}, \tool reaches 42.51\%, significantly outperforming all baselines ($p < 0.001$, McNemar's test).
Meanwhile, the False Refusal Rate on valid inputs is reduced to 1.20\%/0.00\% (Normal/Anony) while the Refusal Rate on blank inputs remains above 92\%.

Beyond the hardware domain, the evaluation methodology and training recipe of \tool carry broader implications for vision-to-code generation as a whole.
The paired Normal/Anony protocol of \textsc{C2VEval} can serve as a general diagnostic for any image-to-code pipeline where textual prompts risk leaking semantic shortcuts, such as CSS class names in UI-to-HTML or axis labels in chart-to-Python tasks.
Likewise, the anonymization, refusal augmentation, and D-ORPO alignment triad provides a transferable recipe for calibrating the hallucination and over-refusal trade-off across visual DSLs, requiring only domain-appropriate identifier masking and negative-pair construction.

In summary, this paper makes four contributions:
\begin{enumerate}
  \item \textbf{Phenomenon.} We identify and document the \emph{Mirage} phenomenon: all eight evaluated MLLMs produce equal or higher scores when the circuit diagram is removed, and genuine visual grounding accounts for only $\sim$8--9\% of samples. 
  \item \textbf{Benchmark \& Evaluation.} We construct \textsc{C2VEval} with exact image-code correspondence and a paired Normal/Anony$\times$Original/Mirage protocol that isolates identifier semantics as the single variable. 
  \item \textbf{Method.} We propose \tool (4B), trained with identifier anonymization, refusal augmentation, and D-ORPO (Decision-Focused ORPO) alignment. 
  \item \textbf{Evaluation.} \tool achieves Functional Pass@1 of 46.11\%/42.51\% (\emph{Normal}/\emph{Anony}) with False Refusal Rates of only 1.20\%/0.00\%, while maintaining $\geq$92\% Refusal Rate on blank inputs. 
\end{enumerate}

To facilitate the replication of {\tool}, we make our source code, trained models, and benchmark publicly available on GitHub.\footnote{\url{https://github.com/NTDXYG/VeriGround}}

The remainder of this paper is organized as follows.
Section~\ref{sec:background} introduces the background and problem definition.
Section~\ref{sec:empirical} presents the empirical study that motivates the approach.
Section~\ref{sec:method} details the proposed method.
Section~\ref{sec:results} reports the experimental results.
Section~\ref{sec:discussion} discusses hyper-parameter sensitivity, mismatch refusal, and threats to validity.
Section~\ref{sec:related} reviews related work.
Section~\ref{sec:conclusion} concludes the paper.
\section{Background and Preliminaries}
\label{sec:background}

\subsection{Multimodal Large Language Models}

A multimodal large language model (MLLM)~\cite{wu2023multimodal, yin2024survey} extends a text-only LLM to process both visual and textual inputs.
Given a visual input $I$ and a textual prompt $T$, an MLLM $\mathcal{M}$ auto-regressively generates an output sequence $\hat{Y} = (y_1, \dots, y_L)$:
\begin{equation}
\label{eq:mllm}
  P(\hat{Y} \mid I, T;\, \theta) = \prod_{t=1}^{L} P(y_t \mid y_{<t}, I, T;\, \theta),
\end{equation}
where $\theta$ denotes the model parameters.
Existing MLLMs broadly fall into two architectural paradigms.

\textbf{Connector-based MLLMs}
adopt a three-component architecture.
A visual encoder (e.g., ViT~\cite{han2022survey, liu2021swin, khan2022transformers}) extracts a feature sequence $\mathbf{z}_v = \mathrm{Enc}(I)$;
a connector (typically an MLP) projects it into the LLM embedding space, yielding visual tokens $\mathbf{h}_v = \mathrm{Proj}(\mathbf{z}_v)$;
an LLM backbone then generates conditioned on $[\mathbf{h}_v;\, \mathbf{h}_t]$, where $\mathbf{h}_t = \mathrm{Embed}(T)$ is the text token embedding sequence.
Representative models include LLaVA~\cite{lillava, li2023llava, cocchi2025llava} and InternVL~\cite{chen2024internvl, wang2025internvl3}.

\textbf{Native multimodal MLLMs}
forgo the connector and jointly train vision and language components within a unified transformer.
A visual encoder produces $\mathbf{h}_v = \mathrm{Enc}(I)$ and a token embedding layer produces $\mathbf{h}_t = \mathrm{Embed}(T)$; both sequences are fed directly into a shared backbone without an intermediate projection, with decoupled parallel strategies for the heterogeneous modalities to maintain training efficiency.
Representative models include Gemini~\cite{team2023gemini, team2024gemini} and GPT-4o~\cite{hurst2024gpt}.

\subsection{Circuit-to-Verilog Code Generation}

We first present the MLLM-based formulation of Circuit-to-Verilog Code Generation.

\begin{definition}[Circuit-to-Verilog Code Generation]
\label{def:task}
Given a circuit diagram image $I$ and a module header $H$, the task is to generate a Verilog module body $\hat{V}$ such that the complete module $H \oplus \hat{V}$ is both syntactically valid and functionally equivalent to the reference implementation $H \oplus V^{*}$, where $\oplus$ denotes concatenation and $V^{*}$ is the ground-truth module body.
The textual prompt takes the form $T = (\mathit{instruction},\, H)$, where the instruction is a fixed task description.
Since the instruction is constant across all samples, we abbreviate the generation as:
\begin{equation}
\label{eq:task}
  \hat{V} = \mathcal{M}(I, H).
\end{equation}
\end{definition}

\noindent
We next detail the two variable inputs to $\mathcal{M}$.

\begin{definition}[Module Header]
\label{def:header}
A \emph{module header} $H$ specifies the external interface of a Verilog module~\cite{navabi1999verilog}:
\begin{equation}
  H \;=\; \texttt{module}\;\mathit{name}\;[\texttt{\#(}\mathit{params}\texttt{)}]\;\texttt{(}\mathit{ports}\texttt{);}
\end{equation}
where $\mathit{name}$ is the module identifier, $\mathit{params}$ is an optional parameter list, and $\mathit{ports}$ specifies input/output ports with their directions and bit-widths.
In standard practice, $H$ is always provided so that the generated interface matches the testbench and evaluation can proceed without name-mismatch failures.
\end{definition}

\begin{definition}[Circuit Diagram]
\label{def:circuit}
A \emph{circuit diagram} $I$ is a visual representation of the target circuit's topology, depicting gates, flip-flops, multiplexers, and their interconnections~\cite{mehler2014digital}.
In our setting, each diagram is rendered from $V^{*}$ via \textsc{netlistsvg}, so the labels in $I$ correspond exactly to the identifiers in $V^{*}$.
\end{definition}

\section{Empirical Study}
\label{sec:empirical}
\subsection{Benchmark Construction}
\label{sec:benchmark}

To investigate the Mirage phenomenon, we construct \textsc{C2VEval}, a circuit-to-Verilog benchmark with exact image-code correspondence.
The construction follows a three-stage pipeline.

\textbf{Stage~1: Seed dataset collection.}
We source problems from four established Verilog code-generation benchmarks: VerilogEval-v2~\cite{liu2023verilogeval, pinckney2025revisiting}, RTLLM-v2~\cite{lu2024rtllm, liu2024openllm}, ResBench~\cite{guo2025resbench}, and ArchXBench~\cite{purini2025archxbench}, which collectively span diverse circuit categories, each accompanied by a complete testbench.

\textbf{Stage~2: Reference code curation.}
For each seed problem, we generate candidate Verilog solutions with two frontier MLLMs (GPT-5.4 and Opus-4.6) and simulate each candidate against the corresponding testbench; only solutions passing all assertions are retained as verified reference implementations $V^{*}$.

\textbf{Stage~3: Diagram rendering and filtering.}
Each $V^{*}$ is rendered into a circuit diagram via \textsc{netlistsvg} and rasterized to JPEG at 96\,DPI, producing image $I$ whose labels correspond exactly to the identifiers in $V^{*}$ (Definition~\ref{def:circuit}).
Samples that fail to render or whose $I$ exceeds $2{,}048$ visual tokens under the Qwen-3.5 tokenizer are discarded.

\subsubsection{Normal and Anony variants}
The above pipeline yields the \emph{Normal} variant of \textsc{C2VEval}.
To test whether high Normal-mode accuracy is a Mirage driven by identifier semantics rather than genuine visual understanding, we further construct an \emph{Anony} (anonymized) variant: 
for each sample, the module name is replaced with a generic \texttt{module\_name}, and all port and parameter identifiers are replaced with positional placeholders \texttt{val\_0}, \texttt{val\_1}, \dots; 
the anonymized code is then re-rendered through the same Stage~3 pipeline to produce a new diagram $I_{\text{anon}}$ paired with an anonymized header $H_{\text{anon}}$.
The circuit topology is preserved; only the semantic cues are removed.
For example, the following Normal-mode header:
\begin{lstlisting}
module sync_fifo #(DEPTH=32, WIDTH=8)
    (clk, rst_n, wr_en, rd_en);
\end{lstlisting}
\vspace{-6pt}
is anonymized to:
\begin{lstlisting}
module module_name #(val_0=32, val_1=8)
    (val_2, val_3, val_4, val_5);
\end{lstlisting}
\vspace{-6pt}
The corresponding circuit diagram is re-rendered with the same positional placeholders, ensuring that the only way to produce correct code is to read the visual topology.

\begin{table}[t]
\small
\centering
\caption{Category distribution of \textsc{C2VEval}.}
\label{tab:benchmark_category}
\begin{tabular}{lrr}
\toprule
\textbf{Category} & \textbf{Count} & \textbf{Ratio} \\
\midrule
Basic Combinational Logic          & 81 & 48.5\% \\
Sequential Building Blocks         & 43 & 25.7\% \\
Finite State Machines              & 32 & 19.2\% \\
Mathematical Operations \& Algo.   & 11 &  6.6\% \\
\midrule
\textbf{Total}                     & \textbf{169} & 100\% \\
\bottomrule
\end{tabular}
\end{table}
    
\subsubsection{Benchmark statistics}
The resulting \textsc{C2VEval} comprises \textbf{169}~samples.
Each sample is a tuple $(I,\,H,\,V^{*},\,\mathcal{T},\,D)$ together with its anonymized counterpart $(I_{\text{anon}},\,H_{\text{anon}})$, where $I$ is the circuit diagram, $H$ is the module header (Definition~\ref{def:header}), $V^{*}$ is the reference module body, $\mathcal{T}$ is the testbench, and $D$ is the natural-language functional description.

Table~\ref{tab:benchmark_category} reports the category distribution: Basic Combinational Logic constitutes the largest share (48.5\%), followed by Sequential Building Blocks (25.7\%) and Finite State Machines (19.2\%), with Mathematical Operations \& Algorithms as the smallest category (6.6\%).
This distribution reflects the natural composition of introductory-to-intermediate RTL design tasks in the seed benchmarks.

\begin{figure}[t]
\centering
\includegraphics[width=0.48\textwidth]{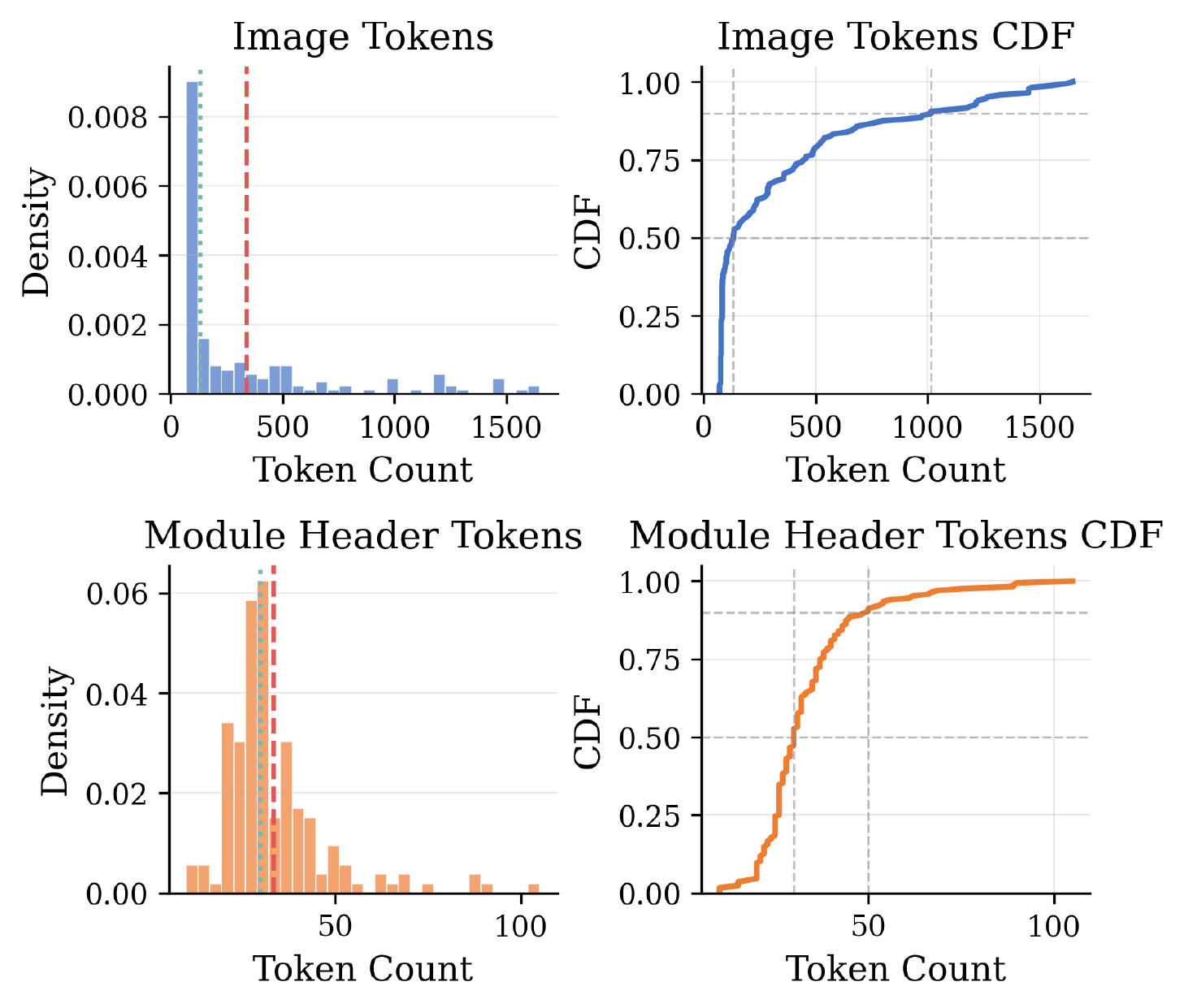}
\caption{Token-count distributions (left: density histograms; right: CDFs) for circuit-diagram images (top, blue) and module headers (bottom, orange) in \textsc{C2VEval}.}
\label{fig:token_dist}
\end{figure}

\begin{table*}[t]
\centering
\small
\caption{Pass@$k$ (\%) on \textsc{C2VEval} under the Normal and Anony benchmark variants.
Mirage mode replaces the circuit diagram with a blank image while the module header is retained.
\textbf{Bold} denotes the higher score between Original and Mirage in each pair.}
\label{tab:main_results}
\begin{tabular}{ll cccc cccc}
\toprule
\multirow{2}{*}{\textbf{Model}}
& \multirow{2}{*}{\textbf{Mode}}
& \multicolumn{4}{c}{\textbf{Normal}}
& \multicolumn{4}{c}{\textbf{Anony}} \\
\cmidrule(lr){3-4} \cmidrule(lr){5-6} \cmidrule(lr){7-8} \cmidrule(lr){9-10}
& & Syn. @1 & Syn. @5 & Func.@1 & Func.@5 & Syn. @1 & Syn. @5 & Func.@1 & Func.@5 \\
\midrule

\rowcolor{mygray}
{GPT-5.4}
& Original & 97.60 & -- & 45.51 & --
            & 93.41 & -- & \textbf{24.55} & -- \\
\rowcolor{mygray}
& Mirage   & \textbf{100.00} & -- & \textbf{47.90} & --
            & \textbf{100.00} & -- & 6.59 & -- \\
{GPT-4o}
& Original & 92.22 & -- & 39.52 & --
            & 79.64 & -- & \textbf{16.17} & -- \\
& Mirage   & \textbf{94.61} & -- & 39.52 & --
            & \textbf{94.01} & -- & 8.38 & -- \\
\rowcolor{mygray}
{MiMo-v2-omni}
& Original & 85.03 & -- & 37.72 & --
            & 79.64 & -- & \textbf{19.16} & -- \\
\rowcolor{mygray}
& Mirage   & \textbf{86.23} & -- & \textbf{41.32} & --
            & \textbf{82.63} & -- & 5.99 & -- \\
{Opus-4.6}
& Original & 98.20 & -- & 52.69 & --
            & 88.62 & -- & 11.38 & -- \\
& Mirage   & \textbf{100.00} & -- & \textbf{63.47} & --
            & \textbf{99.40} & -- & \textbf{14.97} & -- \\
\midrule
\rowcolor{mygray}
{EGM 4B}
& Original & 79.04 & 98.20 & 18.56 & 29.94
            & \textbf{50.90} & \textbf{93.41} & \textbf{4.19} & \textbf{10.78} \\
\rowcolor{mygray}
& Mirage   & \textbf{83.83} & \textbf{99.40} & \textbf{20.36} & \textbf{31.14}
            & 33.53 & 86.83 & 0.60 & 7.78 \\
{EGM 8B}
& Original & 37.72 & 64.67 & 10.18 & 23.35
            & 22.16 & 71.26 & 0.60 & \textbf{8.38} \\
& Mirage   & \textbf{83.83} & \textbf{95.81} & \textbf{22.75} & \textbf{36.53}
            & \textbf{23.95} & \textbf{73.65} & 0.60 & 5.99 \\
\rowcolor{mygray}
{Qwen3.5 4B}
& Original & 38.32 & 76.50 &  7.78 & 19.16
            & 35.33 & 78.44 & \textbf{5.99} & \textbf{12.57} \\
\rowcolor{mygray}
& Mirage   & \textbf{56.89} & \textbf{94.10} & \textbf{10.78} & \textbf{25.75}
            & \textbf{61.68} & \textbf{98.80} & 2.40 & 5.39 \\
{Qwen3.5 9B}
& Original & 70.66 & 97.60 & 14.97 & 32.34
            & 62.28 & 92.22 & \textbf{8.38} & \textbf{14.97} \\
& Mirage   & \textbf{76.05} & \textbf{99.40} & \textbf{19.76} & \textbf{38.92}
            & \textbf{79.04} & \textbf{100.00} & 3.59 & 7.78 \\
\bottomrule
\end{tabular}
\end{table*}

Fig.~\ref{fig:token_dist} visualizes the token-count distributions of the two inputs.
Circuit-diagram images exhibit a heavy right-skewed distribution ($\mu = 340.6$, median~$= 133.0$): approximately half of the diagrams encode within 133 tokens, yet a few complex designs exceed 1{,}500 tokens.
Module headers, by contrast, are remarkably compact ($\mu = 33.5$, median~$= 30.0$), with over 90\% falling below 60 tokens.
This order-of-magnitude asymmetry is noteworthy: the header occupies less than one-tenth of the visual token budget, yet it encodes the module name, port names, and parameter values, which collectively carry rich semantic information.

\subsection{Empirical Setup}
\label{sec:setup}

\subsubsection{Evaluated Models}

We evaluate eight MLLMs spanning 4B to frontier-scale parameters.
The \emph{proprietary} group includes GPT-4o~\cite{openai_gpt4o_docs}, GPT-5.4~\cite{openai_gpt54_docs}, Opus-4.6~\cite{anthropic_claude_opus_46}, and MiMo-v2-omni~\cite{MiMo}.
The \emph{open-source} group includes EGM-4B, EGM-8B~\cite{zhan2026EGM}, Qwen3.5-4B, and Qwen3.5-9B~\cite{qwen3.5}.
Each model is evaluated under both the Normal and Anony variants of \textsc{C2VEval}.

\subsubsection{Evaluation Metrics}

We adopt two Pass@$k$ metrics computed with the unbiased estimator~\cite{chen2021evaluating, yang2026semantic}.
\textbf{Syntax Pass@$k$} measures the probability that at least one of $k$ sampled completions yields a syntactically compilable Verilog module.
\textbf{Functional Pass@$k$} measures the probability that at least one of $k$ completions passes all testbench assertions under simulation.
Functional correctness implies syntactic validity, so Functional Pass@$k$ $\leq$ Syntax Pass@$k$.

\subsubsection{Implementation Details}
Due to API cost constraints, proprietary models are reported as Pass@$1$ only.
Open-source models are served locally; we sample $n = 5$ completions per problem at temperature~$= 0.7$ and report both Pass@$1$ and Pass@$5$.
All compilation and simulation are performed with Icarus Verilog.

\begin{table}[t]
\centering
\small
\caption{Sample-level Functional Pass@1 breakdown (\%) on \textsc{C2VEval} (167 samples). Each sample is categorized by whether Original and Mirage each produce correct code.}
\label{tab:sample_breakdown}
\begin{tabular}{l rrrr}
\toprule
\textbf{Model}
    & \textbf{Both} & \textbf{Original} & \textbf{Mirage} & \textbf{Neither} \\
\midrule
\multicolumn{5}{l}{\textbf{Normal}} \\
\midrule
\rowcolor{mygray}
GPT-5.4      & 35.3 & 10.2 & 12.6 & 41.9 \\
GPT-4o       & 30.5 &  9.0 &  9.0 & 51.5 \\
\rowcolor{mygray}
MiMo-v2      & 26.3 & 11.4 & 15.0 & 47.3 \\
Opus-4.6     & 43.1 &  9.6 & 20.4 & 26.9 \\
\midrule
\rowcolor{mygray}
EGM 4B       & 10.2 &  8.4 & 10.2 & 71.3 \\
EGM 8B       &  7.8 &  2.4 & 15.0 & 74.9 \\
\rowcolor{mygray}
Qwen3.5 4B   &  2.4 &  5.4 &  8.4 & 83.8 \\
Qwen3.5 9B   &  5.4 &  9.6 & 14.4 & 70.7 \\
\midrule
\textbf{Avg.}& \textbf{20.1} & \textbf{8.2} & \textbf{13.1} & \textbf{58.5} \\
\midrule
\multicolumn{5}{l}{\textbf{Anony}} \\
\midrule
\rowcolor{mygray}
GPT-5.4      &  6.0 & 18.6 &  0.6 & 74.9 \\
GPT-4o       &  4.2 & 12.0 &  4.2 & 79.6 \\
\rowcolor{mygray}
MiMo-v2      &  2.4 & 16.8 &  3.6 & 77.2 \\
Opus-4.6     &  4.8 &  6.6 & 10.2 & 78.4 \\
\midrule
\rowcolor{mygray}
EGM 4B       &  0.0 &  4.2 &  0.6 & 95.2 \\
EGM 8B       &  0.0 &  0.6 &  0.6 & 98.8 \\
\rowcolor{mygray}
Qwen3.5 4B   &  0.6 &  5.4 &  1.8 & 92.2 \\
Qwen3.5 9B   &  1.8 &  6.6 &  1.8 & 89.8 \\
\midrule
\textbf{Avg.}& \textbf{2.5} & \textbf{8.8} & \textbf{2.9} & \textbf{85.8} \\
\bottomrule
\end{tabular}
\end{table}

\begin{figure*}[t]
\centering
\includegraphics[width=0.95\textwidth]{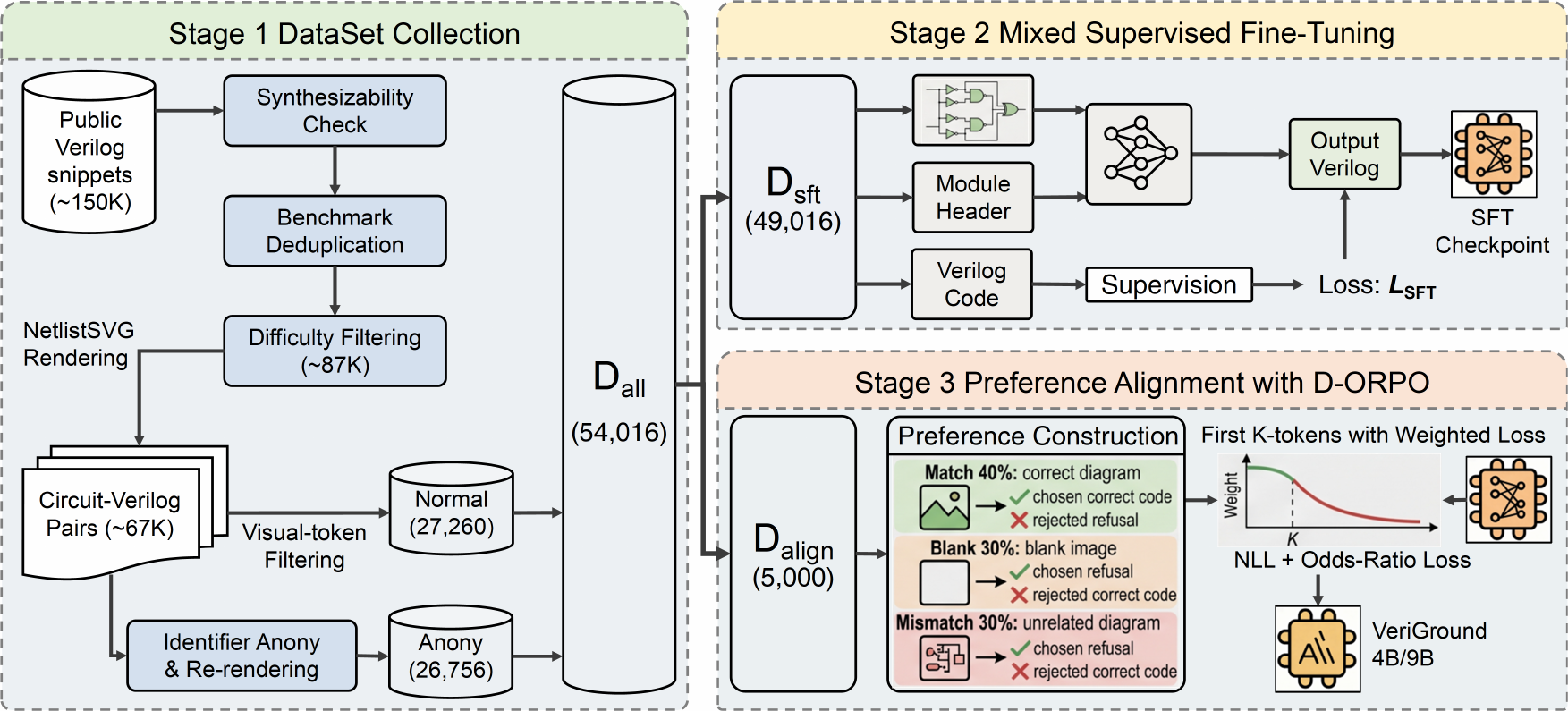}
\caption{Training pipeline of \tool. Stage~1 collects and anonymizes circuit-diagram/Verilog pairs to form a mixed Normal\,+\,Anony corpus. Stage~2 performs supervised fine-tuning on this corpus, forcing the model to ground in visual topology rather than identifier semantics. Stage~3 constructs preference pairs with refusal augmentation and applies D-ORPO alignment to balance generation quality against hallucination and over-refusal.}
\label{fig:pipeline}
\end{figure*}

\subsection{Empirical Findings}
\label{sec:findings}

\begin{findingbox}
Finding 1: The Mirage phenomenon is universal: removing the circuit diagram does not degrade, and often improves, code generation.
\end{findingbox}
Under the Normal benchmark, Mirage mode matches or exceeds Original mode on \emph{every} model across all metrics (Table~\ref{tab:main_results}).
The most extreme case is EGM~8B, whose Syntax Pass@1 jumps by 46 points upon removing the circuit diagram, while Opus-4.6 sees a Functional gain of over 10 points.
The same trend holds at Pass@5: all four open-source models see Mirage surpass Original on every metric, ruling out single-sample variance.
A sample-level decomposition (Table~\ref{tab:sample_breakdown}) corroborates this pattern: on average, 13.1\% of samples are solved by Mirage only versus 8.2\% by Original only, meaning the diagram actively \emph{impairs} generation on 60\% more samples than it aids.

\begin{findingbox}
Finding 2: Identifier semantics in the module header, not visual understanding, are the primary driver of Normal-mode performance.
\end{findingbox}
Anonymization strips semantic identifiers while preserving circuit topology.
Under the Anony benchmark, the Mirage advantage vanishes: Original surpasses Mirage on Functional Pass@1 for seven of eight models (Table~\ref{tab:main_results}), and this reversal extends to Pass@5 for all four open-source models.
At the sample level (Table~\ref{tab:sample_breakdown}), the Both category, i.e., samples solved by both modes, collapses from 20.1\% to 2.5\%, confirming that the vast majority of Normal-mode joint successes were driven by header-based shortcuts rather than diagram comprehension.
Correspondingly, aggregate Functional Pass@1 drops drastically upon anonymization (e.g., GPT-5.4: $45.51 \to 24.55$; Opus-4.6: $52.69 \to 11.38$), further isolating identifier semantics as the dominant factor.

\begin{findingbox}
Finding 3: Genuine visual grounding is extremely limited, accounting for only $\sim$8--9\% of samples.
\end{findingbox}
The Original-only rate in Table~\ref{tab:sample_breakdown}, which captures samples where the diagram is necessary and sufficient for correct generation, remains nearly identical across Normal (8.2\%) and Anony (8.8\%).
This stability suggests that approximately 8--9\% of samples contain visual cues that models can genuinely extract irrespective of whether identifiers carry semantic meaning.
Beyond this small fraction, current MLLMs largely fail to leverage circuit diagrams: under Anony, 85.8\% of samples are solved by neither mode, indicating that once identifier shortcuts are removed, the overwhelming majority of circuit-to-code tasks remain out of reach.
\section{Method}
\label{sec:method}

Figure~\ref{fig:pipeline} illustrates the overall training pipeline of \tool, which consists of three stages: (i)~training data collection and anonymisation, (ii)~supervised fine-tuning (SFT) on the mixed Normal\,+\,Anony corpus, and (iii)~preference alignment with D-ORPO.

\subsection{Dataset Collection}
\label{sec:dataset}

We construct a large-scale circuit-diagram-to-Verilog training corpus through a six-step pipeline. Steps~1--4 follow the data-curation workflow of CodeV-R1~\cite{zhuqimeng}; Steps~5--6 extend it with diagram rendering and visual-token budgeting.

\textbf{Step~1: Source collection.}
Approximately 150K synthesisable Verilog snippets are harvested from public GitHub repositories, spanning a broad spectrum of hardware designs.

\textbf{Step~2: Synthesisability verification.}
Each snippet is compiled with \textsc{Yosys}~\cite{wolf2013yosys} and discarded if synthesis fails, ensuring that every retained sample corresponds to a valid hardware design.

\textbf{Step~3: Benchmark decontamination.}
To prevent data leakage, we remove any sample whose Rouge-L~\cite{lin2004rouge} similarity with the test sets exceeds $0.5$.

\textbf{Step~4: Difficulty filtering.}
Samples for which Qwen2.5-Coder-7B/32B-Instruct can produce a functionally correct solution within five attempts (verified via formal equivalence checking~\cite{zhuqimeng}) are discarded, retaining approximately 87K non-trivial instances.

\textbf{Step~5: Circuit-diagram rendering.}
Each retained reference implementation $V^{*}$ is synthesised into a gate-level netlist and rendered as an SVG schematic via \textsc{netlistsvg}. Samples that fail to render are excluded, leaving approximately 67K Circuit--Verilog pairs.

\textbf{Step~6: Visual-token budget filtering.}
The SVG schematics are rasterised to JPEG at 96\,DPI. Samples whose image exceeds 2{,}048 visual tokens under the Qwen-3.5 tokenizer are discarded, yielding the \textbf{Normal} training set of \textbf{27{,}260} samples.

\textbf{Anonymised variant.}
Applying the same anonymisation procedure used for \textsc{C2VEval} (Section~\ref{sec:benchmark}), we replace all semantically loaded identifiers with positional placeholders and re-render the corresponding circuit diagrams, producing the \textbf{Anony} training set of \textbf{26{,}756} samples. The Normal and Anony sets are merged into a unified corpus $\mathcal{D}_{\mathrm{all}}$ of \textbf{54{,}016} samples.

\subsection{Mixed Supervised Fine-Tuning}
\label{sec:sft}

We partition $\mathcal{D}_{\mathrm{all}}$ into two disjoint subsets: a seed pool $\mathcal{D}_{\mathrm{align}}$ of 5{,}000 samples reserved for the subsequent alignment stage (Section~\ref{sec:dorpo}), and $\mathcal{D}_{\mathrm{sft}} = \mathcal{D}_{\mathrm{all}} \setminus \mathcal{D}_{\mathrm{align}}$ (\textbf{49{,}016} samples) used for supervised fine-tuning.

Each SFT instance is a triple $(I, H, V^{*})$, where $I$ denotes the circuit diagram, $H$ the module header, and $V^{*}$ the reference Verilog module body. The model $\mathcal{M}_\theta$ is optimised with the standard autoregressive objective:
\begin{equation}
  \mathcal{L}_{\mathrm{SFT}} =
    -\sum_{t=1}^{|V^{*}|} \log\, p_\theta(v_t^{*} \mid I, H, v_{<t}^{*}),
  \label{eq:sft}
\end{equation}
where $v_t^{*}$ denotes the $t$-th token of $V^{*}$. 

Because $\mathcal{D}_{\mathrm{sft}}$ interleaves \emph{Normal} and \emph{Anony} samples in roughly equal proportion, the same module header $H$ appears with both semantically loaded and anonymised identifiers across different instances. 
This compels the model to ground its predictions in the visual topology of $I$ rather than memorising identifier-to-template mappings, as a direct countermeasure to the Mirage phenomenon identified in Section~\ref{sec:findings}.

\subsection{Preference Pair Construction}
\label{sec:pairs}

The alignment stage operates on preference pairs $(y_w, y_l)$ (chosen vs.\ rejected), each conditioned on an input $(I, H)$. A well-calibrated circuit-to-Verilog model should produce correct code when the diagram faithfully depicts the target circuit, and \emph{refuse to answer} when the visual evidence is absent or inconsistent with the header. All three categories share the unified prompt and refusal templates shown in Figure~\ref{fig:prompt_template}. We derive three complementary categories of preference pairs from $\mathcal{D}_{\mathrm{align}}$, summarised in Table~\ref{tab:pref_pairs}.

\begin{figure}[t]
\centering
\begin{tcolorbox}[colback=lightgray, colframe=black!50, title={\small\textbf{Prompt Template}}, fonttitle=\bfseries\small, left=4pt, right=4pt, top=2pt, bottom=2pt]
\ttfamily\small
Please write a Verilog module based on the provided circuit diagram image. Return only the Verilog code, without any explanation.\\[4pt]
For example:\\
\texttt{```verilog}\\
\texttt{your Verilog code here}\\
\texttt{```}\\[4pt]
Module header (must not be changed):\\
\texttt{\{module\_header\}}
\end{tcolorbox}
\vspace{2pt}
\begin{tcolorbox}[colback=lightgray, colframe=black!50, title={\small\textbf{Refusal Response Template}}, fonttitle=\bfseries\small, left=4pt, right=4pt, top=2pt, bottom=2pt]
\ttfamily\small
Based on the provided circuit diagram, I cannot accurately determine the Verilog implementation.\\[4pt]
The module header provided is:\\
\texttt{\{module\_header\}}\\[4pt]
However, the provided image does not match the given module header, so I cannot generate the correct Verilog code with confidence.
\end{tcolorbox}
\caption{Prompt and refusal templates used for preference-pair construction. \texttt{\{module\_header\}} is instantiated with the sample-specific header $H$.}
\label{fig:prompt_template}
\end{figure}

\begin{table}[t]
\centering
\footnotesize
\caption{Preference-pair composition for the alignment stage.}
\label{tab:pref_pairs}
\begin{tabular}{clccr}
\toprule
\textbf{Category} & \textbf{Image condition} & \textbf{Chosen}
  & \textbf{Rejected} & \textbf{Ratio} \\
\midrule
\textsc{Match}    & Correct diagram     & Verilog & Refusal      & 40\% \\
\textsc{Blank}    & Blank image & Refusal      & Verilog & 30\% \\
\textsc{Mismatch} & Unrelated diagram   & Refusal      & Verilog & 30\% \\
\bottomrule
\end{tabular}
\end{table}

\paragraph{Category \textsc{Match} (40\%)}
The input pairs a \emph{matching} circuit diagram $I$ with its corresponding module header $H$.
The chosen response $y_w$ is the reference module body $V^{*}$; the rejected response $y_l$ follows the refusal template in Figure~\ref{fig:prompt_template}.
This category reinforces the model's generation capability, ensuring that valid visual evidence is not spuriously refused.

\paragraph{Category \textsc{Blank} (30\%)}
The circuit diagram is replaced with a blank white image $I_{\varnothing}$, while the module header $H$ is retained. The chosen response follows the refusal template; the rejected response is the header-matching code $V^{*}$.
By explicitly penalising code generation in the absence of visual input, this category directly targets the Mirage shortcut.

\paragraph{Category \textsc{Mismatch} (30\%)}
The circuit diagram $I'$ is sampled from a different instance in $\mathcal{D}_{\mathrm{align}}$ such that $I'$ depicts a circuit unrelated to $H$.
Refusal is again chosen over the header-matching code.
This category trains the model to detect semantic inconsistency between the visual and textual modalities, rather than defaulting to header-driven generation.

\textbf{Ratio design.}
Each source sample in $\mathcal{D}_{\mathrm{align}}$ naturally yields exactly one \textsc{Match} pair, one \textsc{Blank} pair, and one \textsc{Mismatch} pair, giving a raw ratio of 1\,:\,1\,:\,1.
We observe that this equal split over-represents refusal-preferred pairs (two out of three categories choose refusal), biasing the model toward over-refusal.
To counterbalance, we up-sample \textsc{Match} to 40\% and assign 30\% each to \textsc{Blank} and \textsc{Mismatch}, yielding a 4\,:\,3\,:\,3 split that keeps the three categories close in size while giving the generation-preferred category a slight majority.

\subsection{Decision-Focused ORPO}
\label{sec:dorpo}


\subsubsection{ORPO Background}
Preference alignment methods such as DPO~\cite{rafailov2023direct} and GRPO~\cite{shao2024deepseekmath} impose significant overhead on multimodal models: DPO doubles memory by maintaining a frozen reference model, while GRPO requires costly online rollouts with high-resolution image encoding.
ORPO~\cite{hong2024orpo} avoids both costs by unifying supervised learning and preference optimisation in a single, reference-free objective.

Given a preference pair $(y_w, y_l)$ conditioned on input $x$, ORPO combines the negative log-likelihood (NLL) of the chosen response with an odds-ratio (OR) penalty:
\begin{equation}
  \mathcal{L}_{\mathrm{ORPO}} = \underbrace{-\frac{1}{|y_w|}\sum_{t} \log p_\theta(y_{w,t} \mid x, y_{w,<t})}_{\mathcal{L}_{\mathrm{NLL}}} \;+\; \beta \cdot \mathcal{L}_{\mathrm{OR}},
  \label{eq:orpo}
\end{equation}
where $\beta$ controls the preference strength and the OR term is
\begin{equation}
  \mathcal{L}_{\mathrm{OR}} = -\log \sigma\!\Big(\log \frac{\mathrm{odds}_\theta(y_w \mid x)}{\mathrm{odds}_\theta(y_l \mid x)}\Big),
  \label{eq:or}
\end{equation}
with $\mathrm{odds}_\theta(y \mid x) \triangleq p_\theta(y \mid x)\,/\,(1 - p_\theta(y \mid x))$.
In practice, the sequence-level log-probability is replaced by its token-level average:
\begin{equation}
  \bar{\ell}_\theta(y \mid x) = \frac{1}{T}\sum_{t=1}^{T} \log p_\theta(y_t \mid x, y_{<t}),
  \label{eq:avg_logp}
\end{equation}
where $T = |y|$.
Note that every token contributes equally to $\bar{\ell}_\theta$; we revisit this assumption below.

\subsubsection{Motivation}
In our generate-or-refuse setting, the binary decision is fully determined by the \emph{first few response tokens} (cf.\ Figure~\ref{fig:prompt_template}): once the model emits \code{```verilog} it is committed to code generation, whereas a refusal opens with \code{Based on the ..., I cannot ...}.
This creates a \emph{length asymmetry}: code responses span hundreds of tokens while refusals contain only tens, so the uniform weighting of Eq.~\ref{eq:avg_logp} dilutes the odds-ratio gradient on the pivotal opening tokens.

We observe that this dilution leads to \emph{over-refusal}: the model favours the shorter refusal path even for valid inputs, because the uniform odds ratio provides insufficient gradient on the initial decision to generate.

\subsubsection{Formulation}
To address this, we propose \textbf{D-ORPO} (\textbf{D}ecision-Focused ORPO), which assigns a higher weight $\alpha > 1$ to the first $K$ response tokens in both the NLL and OR objectives.
Let $r = \min\{t : y_t \text{ is a response token}\}$ denote the index of the first response token.

We define a per-token weight function:
\begin{equation}
  w_t =
  \begin{cases}
    \alpha, & \text{if } r \leq t < r + K, \\
    1, & \text{otherwise},
  \end{cases}
  \label{eq:weight}
\end{equation}
where $K$ is the \emph{decision window size} and $\alpha$ is the \emph{decision weight}.

The weighted average log-probability replaces Eq.~\ref{eq:avg_logp}:
\begin{equation}
  \bar{\ell}^{\,\text{D}}_\theta(y \mid x) = \frac{\sum_{t=1}^{T} w_t \cdot \log p_\theta(y_t \mid x, y_{<t})}{\sum_{t=1}^{T} w_t}.
  \label{eq:weighted_logp}
\end{equation}

Substituting $\bar{\ell}^{\,\text{D}}_\theta$ for $\bar{\ell}_\theta$ in Eqs.~\ref{eq:orpo} and~\ref{eq:or} yields the D-ORPO objective:
\begin{equation}
  \mathcal{L}_{\text{D-ORPO}} = \mathcal{L}^{\,\text{D}}_{\mathrm{NLL}} + \beta \cdot \mathcal{L}^{\,\text{D}}_{\mathrm{OR}}.
  \label{eq:dorpo}
\end{equation}
When $\alpha = 1$, D-ORPO reduces to standard ORPO, confirming it as a strict generalisation.

\subsubsection{Theoretical Analysis}
\begin{definition}[Decision Gradient Fraction]
\label{def:dgf}
For a response of length $T$ with decision window $K \leq T$ and weight $\alpha \geq 1$, the \emph{decision gradient fraction}
\begin{equation}
  \phi(T, K, \alpha) \;\triangleq\; \frac{\alpha K}{\alpha K + (T - K)}
  \label{eq:phi}
\end{equation}
measures the fractional contribution of the $K$ decision tokens to the gradient of $\bar{\ell}^{\,\text{D}}_\theta$ (Eq.~\ref{eq:weighted_logp}).
When $\alpha = 1$, this reduces to $K / T$.
\end{definition}

\begin{proposition}[Gradient Rebalancing]
\label{prop:rebalance}
Let $T_c > T_r > K$ denote the lengths of a code response and a refusal response, respectively.
The \emph{decision gradient imbalance ratio}
\begin{equation}
  \Gamma(\alpha) \;\triangleq\; \frac{\phi(T_r, K, \alpha)}{\phi(T_c, K, \alpha)}
  \;=\; \frac{\alpha K + (T_c - K)}{\alpha K + (T_r - K)}
  \label{eq:gamma}
\end{equation}
satisfies:
\begin{enumerate}
  \item $\Gamma(1) = T_c / T_r$;
  \item $\Gamma(\alpha)$ is strictly decreasing in $\alpha$ for all $\alpha \geq 1$;
  \item $\lim_{\alpha \to \infty} \Gamma(\alpha) = 1$.
\end{enumerate}
\end{proposition}

\begin{proof}
(1) follows by direct substitution.
For (2), $d\Gamma / d\alpha = K(T_r - T_c) / [\alpha K + (T_r - K)]^{2} < 0$ since $T_r < T_c$.
For (3), dividing numerator and denominator by $\alpha K$ gives $\Gamma \to 1$.
\end{proof}

Property~(1) quantifies the imbalance under standard ORPO: the decision gradient fraction for the short refusal is $T_c / T_r$ times that for the long code response, causing the OR objective to steer the decision boundary toward refusal far more strongly than toward generation.
Properties~(2) and~(3) guarantee that D-ORPO monotonically reduces this imbalance toward unity as $\alpha$ increases.
Algorithm~\ref{alg:dorpo} summarises the D-ORPO training procedure.
\begin{algorithm}[t]
\caption{D-ORPO Training}
\label{alg:dorpo}
\KwIn{SFT-initialised model $\mathcal{M}_\theta$; preference dataset $\mathcal{D}_{\mathrm{pref}}$; decision window $K$; decision weight $\alpha$; OR coefficient $\beta$}
\KwOut{Aligned model $\mathcal{M}_\theta$}
\For{each mini-batch $\{(x_i, y_{w}^{(i)}, y_{l}^{(i)})\}_{i=1}^{B}$ from $\mathcal{D}_{\mathrm{pref}}$}{
  Forward pass: compute per-token log-probs $\log p_\theta(y_t \mid x, y_{<t})$ for both $y_w$ and $y_l$\;
  Compute per-token weights $w_t$ via Eq.~\ref{eq:weight}\;
  Compute $\bar{\ell}^{\,\text{D}}_\theta(y_w \mid x)$ and $\bar{\ell}^{\,\text{D}}_\theta(y_l \mid x)$ via Eq.~\ref{eq:weighted_logp}\;
  Compute $\mathcal{L}_{\text{D-ORPO}}$ via Eq.~\ref{eq:dorpo}\;
  Update $\theta$ via back-propagation on $\mathcal{L}_{\text{D-ORPO}}$\;
}
\end{algorithm}

\begin{table*}[t]
\centering
\small
\caption{Pass@$k$ (\%) of \tool ablation variants on \textsc{C2VEval} under the paired Normal/Anony $\times$ Original/Mirage protocol.
For Original rows, \textbf{bold} = highest (best); for Mirage rows, \textbf{bold} = lowest (best).}
\label{tab:rq1_4b}
\begin{tabular}{ll cccc cccc}
\toprule
\multirow{2}{*}{\textbf{Variant}}
& \multirow{2}{*}{\textbf{Mode}}
& \multicolumn{4}{c}{\textbf{Normal}}
& \multicolumn{4}{c}{\textbf{Anony}} \\
\cmidrule(lr){3-4} \cmidrule(lr){5-6} \cmidrule(lr){7-8} \cmidrule(lr){9-10}
& & Syn. @1 & Syn. @5 & Func.@1 & Func.@5 & Syn. @1 & Syn. @5 & Func.@1 & Func.@5 \\
\midrule

\rowcolor{mygray}
{Base (Qwen3.5-4B)}
& Original & 38.32 & 76.50 &  7.78 & 19.16
            & 35.33 & 78.44 &  5.99 & 12.57 \\
\rowcolor{mygray}
& Mirage   & 56.89 & 94.10 & 10.78 & 25.75
            & 61.68 & 98.80 & 2.40 & 5.39 \\
{SFT}
& Original & 87.43 & 97.01 & 28.74 & 46.11
            & 91.02 & \textbf{100.00} & 11.98 & 29.34 \\
& Mirage   & 81.44 & 98.80 & 12.57 & 28.74
            & 86.83 & 100.00 &  2.99 &  4.19 \\
\rowcolor{mygray}
{Anony-mixed SFT}
& Original & \textbf{96.41} & \textbf{98.80} & 43.71 & \textbf{61.08}
            & \textbf{97.01} & 99.40 & 40.72 & 54.49 \\
\rowcolor{mygray}
& Mirage   & 88.02 & 100.00 & 14.37 & 28.14
            & 88.02 & 100.00 &  4.19 & 11.38 \\
{\quad+ORPO}
& Original & 95.81 & \textbf{98.80} & 41.92 & 60.48
            & 95.81 & 99.40 & 40.72 & 51.50 \\
& Mirage   &  8.38 & 34.73 &  0.60 &  4.79
            &  3.59 & 11.98 &  \textbf{0.00} &  \textbf{0.00} \\
\rowcolor{mygray}
{\quad+D-ORPO}
& Original & \textbf{96.41} & \textbf{98.80} & \textbf{46.11} & \textbf{61.08}
            & 96.41 & 99.40 & \textbf{42.51} & \textbf{56.29} \\
\rowcolor{mygray}
& Mirage   & \textbf{3.59} & \textbf{17.37} & \textbf{0.00} & \textbf{2.99}
            & \textbf{1.80} & \textbf{10.78} & \textbf{0.00} & \textbf{0.00} \\
\bottomrule
\end{tabular}
\end{table*}

\subsection{Training Details}
\label{sec:training_details}

We instantiate \textbf{\tool}, initialised from Qwen3.5-4B.
\tool is fine-tuned with LoRA~\cite{hu2022lora} ($r{=}16$, $\alpha_{\text{LoRA}}{=}16$, dropout\,=\,0) applied to all attention and MLP projection matrices in both the vision encoder and the language backbone.
Training proceeds in two sequential stages.

\textbf{Mixed Supervised Fine-tuning.}
The model is trained on $\mathcal{D}_{\mathrm{sft}}$ (49{,}016 samples) for 5~epochs using AdamW (8-bit) with a learning rate of $2 \times 10^{-4}$, and a maximum sequence length of 4{,}096 tokens.
We hold out the last 5{,}000 samples as a validation set, evaluate every 1{,}000 steps, and retain the checkpoint with the lowest validation loss.

\textbf{D-ORPO Alignment.}
Starting from the best SFT checkpoint, we perform one epoch of D-ORPO training on the preference pairs derived from $\mathcal{D}_{\mathrm{align}}$ (Section~\ref{sec:pairs}).
The learning rate is reduced to $5 \times 10^{-6}$.
The D-ORPO-specific hyper-parameters are set to decision window $K{=}8$, decision weight $\alpha{=}2.0$, and OR coefficient $\beta{=}0.1$.
\section{Experimental Results}
\label{sec:results}

We evaluate \tool by answering two research questions:

\begin{itemize}
  \item \textbf{RQ1}: Does \tool improve grounded code generation?
  \item \textbf{RQ2}: Does \tool refuse unreliable visual inputs without over-refusing valid diagrams?
\end{itemize}

\subsection{RQ1: Code Generation Accuracy}
\label{sec:rq1}
\begin{figure}[t]
\centering
\includegraphics[width=0.5\textwidth]{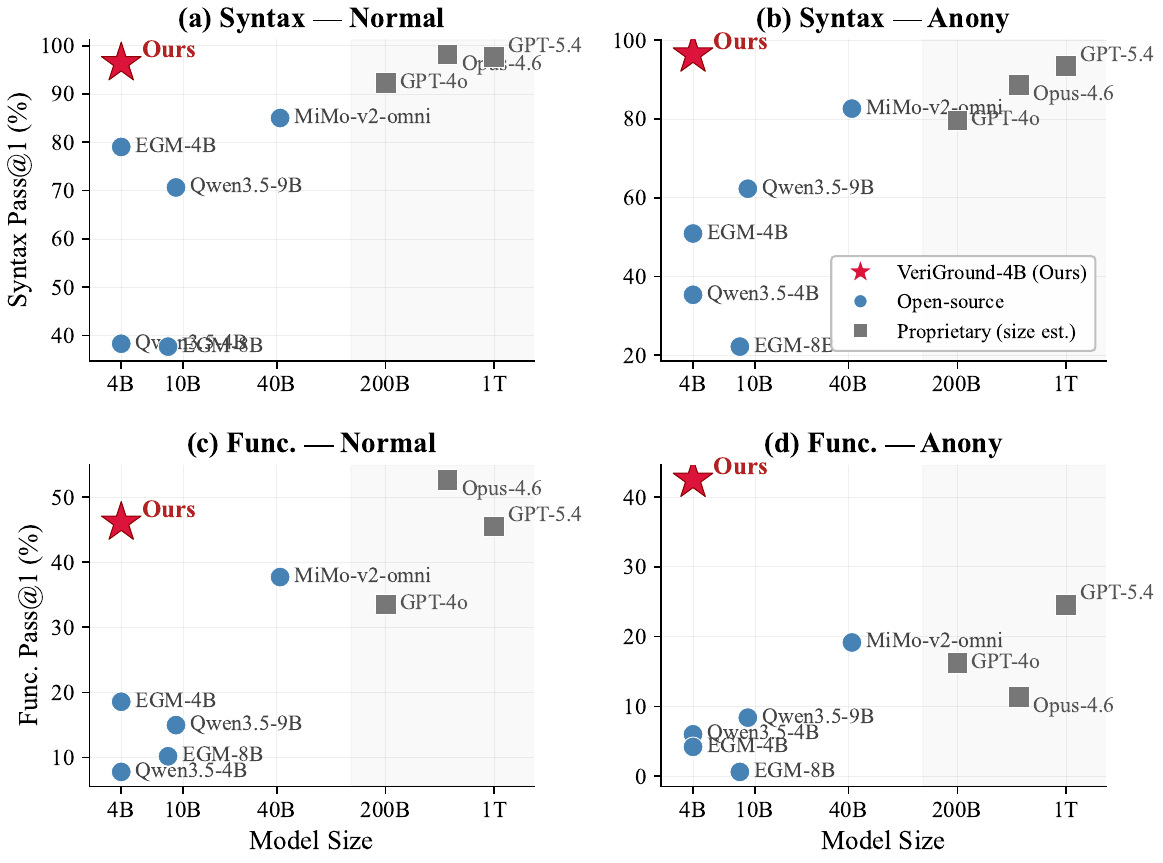}
\caption{Original-mode Pass@1 comparison on \textsc{C2VEval}.
(a--b)~Syntax; (c--d)~Functional.
Open-source models (\textcolor[HTML]{4682B4}{$\bullet$}) are at their
parameter count; MiMo-v2-omni at 42B active (1T total, MoE).
Proprietary models (\textcolor[HTML]{777777}{$\blacksquare$}, size
estimated) are in the shaded region.
\tool (\textcolor[HTML]{DC143C}{$\bigstar$}) approaches GPT-5.4
under Normal and surpasses all models under Anony.}
\label{fig:rq1_comparison}
\end{figure}

To isolate the contribution of each training component, we evaluate four cumulative ablation stages:
(i)~\emph{SFT}, fine-tuned on Normal-only data;
(ii)~\emph{Anony-mixed SFT}, fine-tuned on the merged Normal\,+\,Anony corpus;
(iii)~\emph{+ORPO}, standard ORPO alignment after Anony-mixed SFT;
(iv)~\emph{+D-ORPO} (i.e.\ \tool), decision-focused ORPO alignment after Anony-mixed SFT.
All variants are evaluated on \textsc{C2VEval} under the paired Normal/Anony $\times$ Original/Mirage protocol.

Table~\ref{tab:rq1_4b} reports the ablation results; Figure~\ref{fig:rq1_comparison} compares \tool against all baselines as a function of model size.
Based on these results, we draw three observations.

\textbf{Obs.\,1: Anony-mixed SFT is the key to visual grounding.}
Mixing anonymised training data into SFT yields substantial gains across all settings. 
Compared with Normal-only SFT, Anony-mixed SFT raises Functional Pass@1 from $28.74 \to 43.71$ (+14.97) under Normal-Original, and from $11.98 \to 40.72$ (+28.74) under Anony-Original, nearly closing the Normal--Anony gap ($\Delta$=2.99). 
This confirms that anonymisation training compels the model to ground predictions in visual topology rather than memorising identifier-to-template mappings. 
However, SFT alone cannot teach the model to \emph{refuse} when visual evidence is absent: Mirage Syntax Pass@1 remains as high as 88.02\%, indicating that the model still produces syntactically valid but functionally incorrect code without any circuit diagram.

\textbf{Obs.\,2: D-ORPO achieves the best generation--refusal trade-off.}
Standard ORPO effectively teaches refusal, reducing Mirage Syntax Pass@1 from 88.02\% to 8.38\%/3.59\% (Normal/Anony), but at the cost of generation quality: Functional Pass@1 drops from $43.71 \to 41.92$ (Normal) and $54.49 \to 51.50$ (Func.@5, Anony). 
D-ORPO resolves this tension by further reducing Mirage Syntax Pass@1 to 3.59\%/1.80\% while \emph{simultaneously} lifting Original Functional Pass@1 to 46.11\% (Normal) and 42.51\% (Anony), the highest among all ablation stages. 
The decision-focused weighting concentrates the preference signal on the generate-or-refuse boundary, preventing the over-refusal caused by length asymmetry in standard ORPO.

\textbf{Obs.\,3: A 4B model rivals frontier-scale proprietary MLLMs.}
As shown in Fig.~\ref{fig:rq1_comparison}(c), \tool achieves Functional Pass@1 of 46.11\% under Normal, approaching GPT-5.4 (45.51\%) and surpassing both GPT-4o (33.52\%) and MiMo-v2-omni (37.72\%), despite having $\sim$50--250$\times$ fewer parameters.
Under Anony (Fig.~\ref{fig:rq1_comparison}(d)), the advantage becomes decisive: \tool reaches 42.51\%, exceeding GPT-5.4 (24.55\%) by 17.96 points and Opus-4.6 (11.38\%) by 31.13 points. 
This widening gap under anonymised evaluation confirms that \tool's strength stems from genuine visual grounding rather than identifier-based shortcuts that benefit larger pretrained models.

\begin{table}[t]
\centering
\small
\caption{McNemar's test for Functional Pass@1 between \tool and baseline models on Original inputs.
$b$\,/\,$c$: samples solved only by \tool\,/\,only by the baseline.
$^{**}$\,$p < 0.01$ after Holm--Bonferroni correction; n.s.: not significant.}
\label{tab:mcnemar}
\begin{tabular}{l rrr rrr}
\toprule
& \multicolumn{3}{c}{\textbf{Normal}}
& \multicolumn{3}{c}{\textbf{Anony}} \\
\cmidrule(lr){2-4} \cmidrule(lr){5-7}
\textbf{Baseline}
  & $b$ & $c$ & $p$
  & $b$ & $c$ & $p$ \\
\midrule
MiMo-v2-omni & 32 & 18 & n.s.
              & 48 &  9 & ${<}\,0.001^{**}$ \\
GPT-4o       & 23 & 12 & n.s.
              & 50 &  6 & ${<}\,0.001^{**}$ \\
GPT-5.4      & 22 & 21 & n.s.
              & 42 & 12 & ${<}\,0.001^{**}$ \\
Opus-4.6     & 22 & 33 & n.s.
              & 56 &  4 & ${<}\,0.001^{**}$ \\
\bottomrule
\end{tabular}
\end{table}
  
\textbf{Statistical Significance.}
All models are evaluated on the same 167 \textsc{C2VEval} samples, yielding paired binary outcomes (pass/fail) per sample.
We apply McNemar's test to each \tool-vs-baseline pair: let $b$ denote the number of samples that only \tool solves and $c$ the number that only the baseline solves.
For discordant counts $b{+}c \leq 25$ we use the exact mid-$p$ variant; otherwise the continuity-corrected $\chi^{2}$ form $\chi^{2} = (|b-c|-1)^{2}/(b+c)$.
All eight $p$-values (4~baselines $\times$ 2~conditions) are corrected jointly with the Holm--Bonferroni procedure at $\alpha{=}0.05$.

Table~\ref{tab:mcnemar} reveals a clear dichotomy.
Under \textbf{Normal} evaluation, none of the four comparisons reaches significance: for instance, \tool vs.\ GPT-5.4 yields $b{=}22, c{=}21$ ($p{=}1.0$), indicating virtually identical per-sample outcomes.
This confirms that \tool, with only 4B parameters, performs \emph{on par} with frontier-scale proprietary models when standard signal names are available.
Under \textbf{Anony} evaluation, all four comparisons are highly significant ($p < 0.001$), with $b \gg c$ in every case (e.g.\ $56$ vs.\ $4$ against Opus-4.6).
While baseline models suffer severe performance degradation without meaningful identifiers (Opus-4.6: $52.69\% \to 11.38\%$; GPT-5.4: $45.51\% \to 24.55\%$), \tool retains $42.51\%$---a drop of only 3.60.
This statistically confirms that \tool's accuracy stems from genuine visual grounding rather than identifier-based shortcuts.

\begin{tcolorbox}[width=1.0\linewidth, title={Summary of RQ1}]
Anony-mixed SFT builds strong visual grounding; D-ORPO further lifts Original Functional Pass@1 to 46.11\%/42.51\% (Normal/Anony) while driving Mirage generation to near zero. 
\end{tcolorbox}

\subsection{RQ2: Refusal Reliability}
\label{sec:rq2}

A reliable circuit-to-Verilog model should generate correct code when the diagram faithfully depicts the target circuit and \emph{refuse} otherwise.
We compare +ORPO and +D-ORPO on two complementary metrics, both computed on \textsc{C2VEval}:
\begin{itemize}
  \item \textbf{False Refusal Rate (FRR~$\downarrow$)}: fraction of valid Original-diagram inputs on which the model erroneously refuses.
  \item \textbf{Refusal Rate (RR~$\uparrow$)}: fraction of blank-image (Mirage) inputs on which the model correctly refuses.
\end{itemize}
A well-calibrated model should minimise FRR while maximising RR; we additionally report Functional Pass@1 on Original inputs to confirm that refusal tuning does not degrade generation quality.

\begin{table}[t]
\centering
\small
\caption{Refusal reliability of +ORPO and +D-ORPO on \textsc{C2VEval}.
FRR is measured on Original inputs ($\downarrow$); RR on Mirage inputs ($\uparrow$); Pass@1 on Original inputs ($\uparrow$).
\textbf{Bold} marks the better value per column.}
\label{tab:rq2}
\begin{tabular}{l cc}
\toprule
\textbf{Metric} & \textbf{+ORPO} & \textbf{+D-ORPO} \\
\midrule
FRR (\%) $\downarrow$ -- Normal  &  8.98 & \textbf{1.20} \\
FRR (\%) $\downarrow$ -- Anony   &  7.19 & \textbf{0.00} \\
\midrule
RR (\%) $\uparrow$ -- Normal     & \textbf{98.20} & 92.81 \\
RR (\%) $\uparrow$ -- Anony      & \textbf{98.80} & 97.60 \\
\midrule
Func.\ Pass@1 (\%) $\uparrow$ -- Normal & 44.31 & \textbf{46.11} \\
Func.\ Pass@1 (\%) $\uparrow$ -- Anony  & 38.32 & \textbf{42.51} \\
\bottomrule
\end{tabular}
\end{table}

Table~\ref{tab:rq2} shows that standard ORPO achieves near-perfect RR ($\geq$98\%) but suffers from severe over-refusal: FRR reaches 8.98\% (Normal) and 7.19\% (Anony), meaning the model spuriously rejects roughly one in twelve valid inputs. 
D-ORPO resolves this imbalance decisively: FRR drops to 1.20\% under Normal and \emph{zero} under Anony, a reduction of 7.78 and 7.19 respectively. 
RR decreases only moderately (98.20\% $\to$ 92.81\% Normal; 98.80\% $\to$ 97.60\% Anony), remaining above 92\% across all settings and preserving reliable refusal capability. 
Crucially, the FRR reduction does not come at the expense of generation quality: D-ORPO \emph{improves} Functional Pass@1 from 44.31\% $\to$ 46.11\% (Normal) and 38.32\% $\to$ 42.51\% (Anony), confirming that decision-focused weighting prevents over-refusal from eroding the code-generation pathway.

\begin{tcolorbox}[width=1.0\linewidth, title={Summary of RQ2}]
D-ORPO reduces False Refusal Rate by up to 7.78 compared with standard ORPO (to 1.20\%/0.00\% Normal/Anony) while retaining $\geq$92\% Refusal Rate on blank images and simultaneously improving Functional Pass@1.
\end{tcolorbox}

\section{Discussion}
\label{sec:discussion}

\subsection{Hyper-parameter Sensitivity}
\label{sec:rq3}

D-ORPO introduces two hyper-parameters beyond standard ORPO: the decision window size $K$ and the decision weight $\alpha$.
We conduct ablation studies on \tool-4B, evaluating Functional Pass@1 on Original inputs, False Refusal Rate (FRR) on Original inputs, and blank-image Refusal Rate (RR) under both Normal and Anony conditions.

\subsubsection{Effect of Decision Window Size $K$}

Table~\ref{tab:ablation_K} fixes $\alpha{=}2.0$ and varies $K \in \{2, 4, 8, 16\}$.

\begin{table}[t]
\centering
\small
\caption{Effect of decision window $K$ ($\alpha{=}2.0$ fixed) on \tool-4B.
\textbf{Bold} = best per column.}
\label{tab:ablation_K}
\begin{tabular}{r cc cc cc}
\toprule
\multirow{2}{*}{$K$}
  & \multicolumn{2}{c}{\textbf{Pass@1}~$\uparrow$}
  & \multicolumn{2}{c}{\textbf{FRR}~$\downarrow$}
  & \multicolumn{2}{c}{\textbf{RR}~$\uparrow$} \\
\cmidrule(lr){2-3} \cmidrule(lr){4-5} \cmidrule(lr){6-7}
  & Norm. & Anon. & Norm. & Anon. & Norm. & Anon. \\
\midrule
 2  & 45.51 & 40.12 & 1.80 & 0.60 & 88.62 & 96.41 \\
 4  & \textbf{46.11} & \textbf{42.51} & \textbf{1.20} & 1.80 & 89.22 & 97.01 \\
 8  & \textbf{46.11} & \textbf{42.51} & \textbf{1.20} & \textbf{0.00} & \textbf{92.81} & \textbf{97.60} \\
16  & 41.32 & \textbf{42.51} & 4.19 & \textbf{0.00} & 84.43 & 95.81 \\
\bottomrule
\end{tabular}
\end{table}

$K{=}8$ achieves the best overall balance.
When $K$ is too small ($K{=}2$), the decision window covers too few tokens to capture the full decision preamble, yielding lower Pass@1 under Anony (40.12\%) and reduced RR under Normal (88.62\%).
When $K$ is too large ($K{=}16$), the window extends beyond the decision boundary into code-body tokens, causing Normal Pass@1 to drop to 41.32\% and FRR to rise to 4.19\%.
$K{=}8$ simultaneously achieves the highest Pass@1, the lowest FRR, and the highest RR across both conditions.

\subsubsection{Effect of Decision Weight $\alpha$}

Table~\ref{tab:ablation_alpha} fixes $K{=}8$ and varies $\alpha \in \{1, 2, 3, 5, 10, 20\}$.
Note that $\alpha{=}1$ recovers standard ORPO.

\begin{table}[t]
\centering
\small
\caption{Effect of decision weight $\alpha$ ($K{=}8$ fixed) on \tool-4B.
$\alpha{=}1$ corresponds to standard ORPO.
\textbf{Bold} = best per column.}
\label{tab:ablation_alpha}
\begin{tabular}{r cc cc cc}
\toprule
\multirow{2}{*}{$\alpha$}
  & \multicolumn{2}{c}{\textbf{Pass@1}~$\uparrow$}
  & \multicolumn{2}{c}{\textbf{FRR}~$\downarrow$}
  & \multicolumn{2}{c}{\textbf{RR}~$\uparrow$} \\
\cmidrule(lr){2-3} \cmidrule(lr){4-5} \cmidrule(lr){6-7}
  & Norm. & Anon. & Norm. & Anon. & Norm. & Anon. \\
\midrule
 1 & 44.31 & 38.32 & 8.98 & 7.19 & \textbf{98.20} & \textbf{98.80} \\
 2  & \textbf{46.11} & 42.51 & \textbf{1.20} & \textbf{0.00} & 92.81 & 97.60 \\
 3  & 44.31 & 41.92 & 1.80 & \textbf{0.00} & 88.02 & 97.01 \\
 5  & 42.51 & 41.92 & 1.80 & 1.20 & 80.24 & 97.01 \\
10  & 38.92 & \textbf{43.11} & 3.59 & \textbf{0.00} & 77.84 & 92.22 \\
20  & 44.91 & 41.92 & \textbf{1.20} & 2.40 & 75.45 & 95.81 \\
\bottomrule
\end{tabular}
\end{table}

The results reveal a clear trade-off.
Standard ORPO ($\alpha{=}1$) achieves near-perfect RR ($\geq$98\%) but suffers from severe over-refusal (FRR: 8.98\%/7.19\%), confirming the gradient imbalance characterised by Proposition~\ref{prop:rebalance}.
Setting $\alpha{=}2$ reduces FRR by 7.78 and 7.19 absolute points while maintaining RR above 92\%, and simultaneously lifts Pass@1 to 46.11\%/42.51\%.
Beyond $\alpha{=}2$, Normal RR declines monotonically ($88.02\% \to 80.24\% \to 77.84\% \to 75.45\%$ for $\alpha{=}3, 5, 10, 20$), while Pass@1 shows no further improvement.
$\alpha{=}2$ therefore provides the Pareto-optimal balance and is adopted as the default configuration.

In summary, $K{=}8$ and $\alpha{=}2.0$ yield the best trade-off among generation quality, false refusal, and blank-image refusal.
D-ORPO is robust to moderate perturbations of either hyper-parameter but degrades when $K$ or $\alpha$ deviates substantially from these values.

\subsection{Mismatch Refusal Analysis}
\label{sec:mismatch}

The RR metric in Section~\ref{sec:rq2} measures refusal on \emph{blank} images.
In practice, a model may also encounter \emph{mismatched} diagrams, where the provided circuit image depicts a different module than the one specified in the header.
To evaluate this scenario, we randomly construct five mismatched datasets: in each round, every test sample is paired with a circuit diagram drawn from a different instance in \textsc{C2VEval}.
We compare +ORPO and +D-ORPO on Functional Pass@1 and Mismatch Refusal Rate (MRR).

\begin{table}[t]
\centering
\small
\caption{Functional Pass@1 (Func., \%) and Mismatch Refusal Rate (MRR, \%) of +ORPO and +D-ORPO across five rounds of random image mismatch.}
\label{tab:mismatch}
\begin{tabular}{l cc cc}
\toprule
& \multicolumn{2}{c}{\textbf{+ORPO}} & \multicolumn{2}{c}{\textbf{+D-ORPO}} \\
\cmidrule(lr){2-3} \cmidrule(lr){4-5}
\textbf{Round} & Func. & MRR & Func. & MRR \\
\midrule
1 & 0.00 & 100.00 & 0.60 & 80.84 \\
2 & 0.00 & 100.00 & 1.20 & 82.04 \\
3 & 0.00 & 100.00 & 0.60 & 80.84 \\
4 & 0.00 & 100.00 & 0.60 & 83.23 \\
5 & 0.00 &  99.40 & 1.20 & 78.44 \\
\midrule
Avg. & 0.00 & 99.88 & 0.84 & 81.08 \\
\bottomrule
\end{tabular}
\end{table}

Table~\ref{tab:mismatch} shows that +ORPO refuses nearly all mismatched inputs (MRR 99.88\%), consistent with its over-refusal tendency observed in Section~\ref{sec:rq2}.
+D-ORPO exhibits a lower MRR of 81.08\%, meaning roughly one-fifth of mismatched samples are not explicitly refused.
However, the Functional Pass@1 of +D-ORPO on these mismatched inputs is only 0.84\%, virtually zero, indicating that the non-refused outputs still fail functional verification.

In other words, when +D-ORPO does not refuse a mismatched diagram, it does not hallucinate a ``correct'' implementation either; instead, the generated code is functionally incorrect and would be caught by standard testbench simulation.
This result suggests that +D-ORPO's lower refusal rate does not introduce a meaningful safety risk: the model may occasionally attempt generation on mismatched inputs, but the resulting code is unlikely to pass downstream verification, providing a practical safety net against silent misuse.

A promising direction for future work is to further improve the mismatch refusal rate without increasing the false refusal rate on valid inputs, for example through contrastive training on hard negative pairs where the diagram and header share partial but inconsistent semantics.


\subsection{Generalizability to Other Vision-to-Code Tasks}
\label{sec:generalizability}

The Mirage phenomenon reflects a general vulnerability of vision-to-code pipelines: whenever the textual prompt carries sufficient semantic cues, the model can bypass the visual input and still produce superficially correct code.
We argue that the three-pronged recipe underlying \tool, namely identifier anonymization, refusal augmentation, and D-ORPO alignment, addresses this vulnerability at a level that is largely task-agnostic.

\textbf{Anonymization generalizes to other visual DSLs.}
In UI-to-HTML generation, CSS class names and element \texttt{id}s (e.g., \texttt{nav-bar}, \texttt{login-btn}) can leak layout semantics in the same way that Verilog port names leak circuit functionality; replacing them with positional placeholders would force the model to derive structure from the mockup image.
Similarly, in chart-to-Python tasks, axis labels and legend entries often reveal the data transformation logic; anonymizing them would compel the model to read the visual encoding to reconstruct the plotting script.
The key insight is that any visual DSL whose textual interface contains semantically loaded identifiers is susceptible to shortcut exploitation, and identifier anonymization provides a systematic countermeasure.

\textbf{Refusal augmentation and D-ORPO are task-agnostic.}
Refusal augmentation requires only the ability to construct invalid input pairs, which is straightforward in any vision-to-code setting: one can pair a UI mockup with an unrelated HTML skeleton, or a chart image with a mismatched data-processing prompt.
D-ORPO's mechanism of up-weighting the first $K$ response tokens to sharpen the generate-or-refuse decision boundary operates purely on the response-prefix distribution and is independent of the downstream code language.
Together, the two components offer a transferable framework for calibrating the hallucination and over-refusal trade-off without task-specific architectural changes.

\subsection{Threats to Validity}
\label{sec:threats}

\textbf{Internal validity.}
Our ablation fixes one hyper-parameter while varying the other; joint interactions beyond the tested grid ($K \in \{2,4,8,16\}$, $\alpha \in \{1,2,3,5,10,20\}$) may exist.
All training runs use a single random seed; we mitigate this by reporting Pass@$k$ with the unbiased estimator over $n{=}5$ samples per problem.

\textbf{External validity.}
\textsc{C2VEval} covers 169 Verilog problems drawn from four public benchmarks, spanning combinational logic, sequential blocks, FSMs, and arithmetic modules.
Whether the findings generalise to other hardware description languages (e.g.\ VHDL, SystemVerilog) or larger circuit scales.
We instantiate \tool at 4B parameters; scaling behaviour at larger or smaller model sizes is unexplored.

\textbf{Construct validity.}
Functional Pass@$k$ relies on testbench simulation, which may not cover all corner cases.
Refusal detection is based on template matching against the fixed refusal format (Figure~\ref{fig:prompt_template}); free-form refusals or hedged responses are not captured by this metric.

\section{Related Work}
\label{sec:related}

\subsection{Vision-to-Code Generation}

The broader vision-to-code paradigm treats visual artifacts as \emph{visual domain-specific languages} (visual DSLs) that must be faithfully translated into executable code.
Representative tasks include UI-to-HTML synthesis~\cite{feng2021auto, xiao2024prototype2code, gui2025uicopilot, zhou2025declarui, xiao2025interaction2code, wan2025divide}, where a mockup encodes layout and interaction semantics, and Plot-to-Python generation~\cite{wu2025plot2code, zhao2025chartcoder, ouyang2026dscodebench}, where a chart encodes data transformations and rendering logic.
A common, yet underexplored, risk across these tasks is that models may exploit textual metadata (e.g., axis labels, component names) rather than genuinely parse the visual content, producing outputs that appear correct on standard benchmarks but fail on unseen inputs.

Circuit-to-Verilog generation extends this paradigm to the hardware domain, where a circuit diagram serves as a visual DSL encoding timing, topology, and bit-level semantics.
In the text-only setting, LLMs have been extensively applied to Verilog generation.
Yang et al.~\cite{yang2025large} and Pan et al.~\cite{pan2025survey} provide comprehensive surveys.
On the benchmarking side, VerilogEval~\cite{liu2023verilogeval} and its updated version~\cite{pinckney2025revisiting} evaluate functional correctness via testbench simulation; RTLLM~\cite{lu2024rtllm} and OpenLLM-RTL~\cite{liu2024openllm} target design-level RTL; ResBench~\cite{guo2025resbench} adds resource-awareness for FPGA designs; and ArchXBench~\cite{purini2025archxbench} covers complex digital subsystems.
On the model side, QiMeng-CodeV-R1~\cite{zhuqimeng} augments Verilog generation with chain-of-thought reasoning, while Yang et al.~\cite{yang2026semantic} address backdoor attacks through semantic consensus decoding.
All of these works operate in a text-only setting.

Only a limited number of studies have explored \emph{visual} inputs for Verilog generation.
VGV~\cite{wong2024vgv} prompts MLLMs with circuit screenshots and reports encouraging Thinking Vision methods, but evaluates exclusively under semantically loaded identifiers, leaving the question of genuine visual grounding open.
MGEMMV~\cite{zhang2026mgemmv} proposes a specialised framework for GEMM-oriented circuits, yet likewise does not test whether models exploit header semantics rather than reading the diagram.

Our work is the first to systematically expose and quantify this shortcut across multiple MLLMs, and to provide evidence that the reliability gap observed here reflects a general vulnerability of vision-to-code pipelines, with circuit diagrams representing the most extreme case due to the invisible and safety-critical nature of RTL semantics.

\subsection{Visual Grounding in Multimodal LLMs}

Whether MLLMs genuinely rely on visual inputs has attracted growing scrutiny.
Asadi et al.~\cite{asadi2026mirage} coin the term ``Mirage'' to describe a broad failure mode in which MLLMs appear to understand images but instead exploit spurious correlations or textual priors.
Their study spans medical imaging, scientific figures, and general-purpose vision-language tasks, demonstrating that performance often persists after visual information is removed or corrupted.

\subsection{Preference Alignment for Language Models}

Aligning language models with human preferences has evolved from reinforcement learning from human feedback (RLHF) to more efficient offline methods.
DPO~\cite{rafailov2023direct} eliminates the reward model by directly optimising a preference objective, but requires a frozen reference model that doubles memory for large multimodal architectures.
GRPO~\cite{shao2024deepseekmath} employs group-relative policy optimisation with online rollouts, achieving strong results in mathematical reasoning but incurring high sampling costs for vision-language models.
ORPO~\cite{hong2024orpo} unifies supervised learning and preference alignment in a single reference-free objective, making it well suited to resource-constrained multimodal training.

\section{Conclusion}
\label{sec:conclusion}

This paper reveals the \emph{Mirage} phenomenon in circuit-to-Verilog code generation: all eight evaluated MLLMs score equally or higher when the circuit diagram is replaced by a blank image, exploiting identifier semantics in the module header rather than the visual input.
This exposes a covert defect type in AI code generation that threatens MLLMs' trustworthiness.
We construct \textsc{C2VEval} with paired Normal/Anony variants, showing that genuine visual grounding accounts for only about 8\% of samples.
To restore grounding, we propose \tool (4B), combining identifier anonymization, refusal augmentation, and D-ORPO alignment, a transferable recipe for balancing hallucination and over-refusal.
\tool achieves 46.11\%/42.51\% Functional Pass@1 (Normal/Anony), approaching GPT-5.4 under Normal and outperforming all baselines under Anony, while keeping false refusal below 1.20\% and blank-image refusal above 92\%.

In the future, we plan to scale \tool to both lighter and larger architectures and to expand \textsc{C2VEval} with industrial-scale circuits.
We also intend to validate the transferability of the anonymization-refusal-alignment paradigm on other vision-to-code tasks to further establish its generality.

\section*{Acknowledgements}

This work was supported by National Key R\&D Program of China (No. 2024YFB4506400).
The authors would like to thank the editors and the anonymous reviewers for their insightful comments and suggestions, which can substantially improve the quality of this work.
\bibliographystyle{IEEEtran}
\bibliography{main}
 
\end{document}